\documentclass[10pt]{article}
\usepackage[T1]{fontenc}
\usepackage[utf8]{inputenc}
\usepackage[a4paper,width=150mm,top=25mm,bottom=25mm]{geometry}

\usepackage{macros}
\usepackage{thmtools,thm-restate}

\bibliography{refs}

\def\*#1{\boldsymbol{#1}} % Use \*A for \mathbf{A}
\def\+#1{\mathcal{#1}} % Use \+A for \mathcal{A}
\def\-#1{\mathrm{#1}} % Use \-A for \mathrm{A}
\def\=#1{\mathbb{#1}} % Use \^A for \mathbb{A}
\def\!#1{\mathfrak{#1}} % Use \!A for \mathfrak{A}

\newcommand{\inner}[2]{\left\langle #1, #2\right\rangle}
\newcommand{\tp}[1]{\left(#1\right)}

\newcommand{\Ker}{\mathsf{Ker}\;}
\newcommand{\Ran}{\mathsf{Ran}\;}

\usepackage{xifthen}
\def\oPr{\mathbf{Pr}}
\renewcommand{\Pr}[2][]{ \ifthenelse{\isempty{#1}}
  {\oPr\left[#2\right]}
  {\oPr_{#1}\left[#2\right]} } % Use \Pr[a]{b} for \mathbf{Pr}_a[b], \Pr{b} for  \mathbf{Pr}[b]

\def\oE{\mathbb{E}}
\renewcommand{\E}[2][]{ \ifthenelse{\isempty{#1}}
  {\oE\left[#2\right]}
  {\oE_{#1}\left[#2\right]} }

\usepackage{xparse}

\def\oVar{\mathbf{Var}}
\renewcommand{\Var}[2][]{ \ifthenelse{\isempty{#1}}
{\oVar\left[#2\right]}
{\oVar_{#1}\left[#2\right]} }
\title{Spectral Gap of Down-Up Walks via Trickle-Down: A Simplified and Sharpened Analysis}
\author{Xiaoyu Chen\thanks{Email: \texttt{xiaoyu@mit.edu}, 
Massachusetts Institute of Technology. Supported by the NSF CAREER grant CCF-2443045, and the Reed Fund at MIT.} \and Kuikui Liu\thanks{Email: \texttt{liukui@mit.edu}, 
Massachusetts Institute of Technology. Supported by the NSF CAREER grant CCF-2443045, and the Reed Fund at MIT.}}
\date{}

\begin{document}
\maketitle

\begin{abstract}
Local-to-global techniques for establishing spectral gaps have played a central role in the modern theory of Markov chain mixing times and the theory of high-dimensional expanders. One of the most striking results in this burgeoning literature is that a spectral gap for the global down-up walk on the facets of a pure simplicial complex can be reduced to sufficiently strong spectral expansion of just the codimension-$2$ links of the complex, a phenomenon colloquially referred to as ``trickle-down''. These types of theorems have had many important applications, including rapid mixing of the exchange walk on the bases of any matroid. In this primarily expository article, we give streamlined proofs of two such theorems in the literature, one by Oppenheim \cite{Oppenheim2018} and one by Leake--Oveis Gharan \cite{LeakeOveisGharan2025}, via an integrated Bochner method. Moreover, in the latter setting, we quantitatively strengthen the dependence of the global spectral gap on the dimension of the complex and the spectral influence, resolving an open question of Leake--Oveis Gharan \cite{LeakeOveisGharan2025}.

\textbf{Disclaimer:} The proofs were developed through a couple rounds of interaction with GPT-5.6 Sol Ultra. We later discovered that \cite{GuoZhang2026Planar} had independently proven the same strengthening of the Leake--Oveis Gharan trickle-down theorem using an extremely similar argument, also found by GPT-5.6 Sol Ultra. The focus of their paper is the complexity of approximating the partition function of spin systems on planar graphs, not on the trickle-down phenomenon itself. In contrast, our motivation is primarily expository, and we hope to bring Bochner-type methods and their connections with the trickle-down phenomenon to the attention of a wider community of researchers.
\end{abstract}

%\newpage
%{\small
%\tableofcontents
%}
%\newpage

\section{Introduction}\label{sec:intro}
The study of spectral gaps of random walks on high-dimensional objects has a long history, with deep connections to the analysis of Markov chains, pseudorandomness, group theory, and beyond. Over the past decade, a powerful new paradigm for establishing spectral gaps has been developed based on \emph{local-to-global principles}: One seeks to certify a spectral gap for the global object by only inspecting local pieces of it. An important instantiation of this theme is the ``trickle-down'' phenomenon originating from the theory of high-dimensional expanders. Theorems of this kind have had abundant applications in the construction of expander graphs and their hypergraph/simplicial complex analogs \cite{EK16, KKL16, KO18, OP22, KO23, DD24, HS24, OS24}, as well as in the design and analysis of sampling algorithms for high-dimensional probability distributions that arise in statistical physics and combinatorics \cite{AnariEtAl2019, ALG2021, AbdolazimiOveisGharan2023, AnariKoehlerVuong2024, WZZ24, LeakeLindbergOveisGharan2025, LeakeOveisGharan2025, CCCYZ25}. In this short note, we give streamlined proofs of two notable ``trickle-down'' theorems in the literature, and quantitatively strengthen one of them. To state them, let us first set up the requisite notation regarding random walks on pure simplicial complexes.

% TODO: cite impact of trickle down in combinatorics vis-a-vis Lorentzian polynomials

\paragraph{The Set-Up} Recall that an \emph{(abstract) simplicial complex} $X$ is a downwards-closed family of subsets of a finite universe $U$. We call the elements of $X$ its \emph{faces}, and its inclusion-wise maximal faces its \emph{facets}. We say $X$ is \emph{pure} if all its facets have the same cardinality $n$. We stratify $X$ by ``dimension'': $X(k) \defeq \{\tau \in X : \abs{\tau} = k\}$ for $0 \leq k \leq n$. The \emph{codimension} of a face $\tau$ is the quantity $\mathsf{codim}(\tau) \defeq n - \abs{\tau}$. For a face $\tau \in X$, its \emph{link} $X^{\tau}$ is the simplicial complex on vertex set $U \setminus \tau$ with faces $\{\sigma \setminus \tau : \sigma \in X, \sigma \supseteq \tau\}$.

Now, let $\mu$ be a probability measure with full support over the facets of $X$, and assume $n\geq2$. We call the pair $(X,\mu)$ a weighted simplicial complex. For each $0 \leq k \leq n$, the distribution $\mu = \mu_{n}$ naturally induces a \emph{marginal distribution} $\mu_{k}$ on $X(k)$ via $\mu_{k}(\tau) = \frac{1}{\binom{n}{k}} \cdot \Pr[\sigma \sim \mu]{\sigma \supseteq \tau}$. For a face $\tau$, we write $\mu^{\tau}$ for the conditional measure induced on $\{\sigma \in X(n) : \sigma \supseteq \tau\}$. Naturally, by removing $\tau$ from each $\sigma \in \supp(\mu^{\tau})$, we see that $\mu^{\tau}$ can be viewed simultaneously as a probability measure on the facets of the link $X^{\tau}$.

Given a weighted simplicial complex $(X,\mu)$, the \emph{down-up walk} $P_{n}^{\bigtriangledown}$ is the Markov chain on facets whose evolution is described as follows: If the current state is $\sigma \in X(n)$, then
\begin{enumerate}
    \item remove a uniformly random element $x \in \sigma$, and
    \item add a random element $y \in X^{\sigma \setminus \{x\}}(1)$ with probability $\mu_{1}^{\sigma \setminus \{x\}}(y)$.
\end{enumerate}
The down-up walk is reversible with stationary measure equal to $\mu$. Moreover, it is positive semidefinite with respect to $L^{2}(\mu)$, the space of functions $f : X(n) \to \R$ endowed with the inner product $\langle f,g \rangle_{\mu} \defeq \E[\sigma \sim \mu]{f(\sigma)g(\sigma)}$.

Our eventual goal will be to establish a global spectral gap for the down-up walk $P_{n}^{\bigtriangledown}$ by relating it to certain local random walks which we now describe. The \emph{$1$-skeleton} of $X$ is the graph whose vertex set is $U = X(1)$ and whose edge set is $X(2)$. Together with the measure $\mu$, we can define the \emph{local random walk} on this graph via
\begin{align*}
    P^{\emptyset}(x \to y) &\defeq \frac{1}{2} \cdot \frac{\mu_{2}(x,y)}{\mu_{1}(x)}, \qquad \forall x \neq y \\
    P^{\emptyset}(x \to x) &\defeq 0, \qquad\qquad\qquad\,\, \forall x \in U.
\end{align*}
As we will see in a moment, this random walk is intimately related to the down-up walk $P_{2}^{\bigtriangledown}$ on the complex formed by truncating $X$ to sets of cardinality at most $2$. For any face $\tau \in X$ with codimension at least $2$ (so that the $1$-skeleton graph is well-defined), we write $P^{\tau}$ for the local random walk with respect to the $1$-skeleton of the weighted link $(X^{\tau},\mu^{\tau})$; this random walk is reversible with respect to the conditional marginal measure $\mu_{1}^{\tau}$ on $U \setminus \tau$.

As a special case, we can recover more familiar objects by considering the partite setting. We say $X$ is \emph{$n$-partite} if there exists a partition $U = U_{1} \sqcup \dotsb \sqcup U_{n}$ such that every facet $\sigma \in X(n)$ has exactly one element from each part. In this case, we can equivalently view each facet $\sigma$ as a member of the product space $\prod_{i=1}^{n} U_{i}$, where the $i$\textsuperscript{th} coordinate $\sigma(i)$ is given by the unique element in $\sigma \cap U_{i}$. In this case, $P_{n}^{\bigtriangledown}$ is exactly the \emph{(heat-bath) Glauber dynamics}, where in each step, the chain selects a uniformly random coordinate and resamples it conditioned on the current assignments of all remaining coordinates.

\subsection{Trickle-Down Theorems}
One of the main insights in the local spectral theory of high-dimensional simplicial complexes is that the spectral gap of $P_{n}^{\bigtriangledown}$ can be lower bounded in terms of the spectral gaps for the local random walks $P^{\tau}$. To be more precise, for a reversible Markov chain described by a transition probability matrix $P$ with eigenvalues $1 = \lambda_{1}(P) \geq \lambda_{2}(P) \geq \dotsb$, define its spectral gap as $\gamma(P) \defeq 1 - \lambda_{2}(P)$. A well-known theorem of \cite{AlevLau2020}, quantitatively strengthening prior results of \cite{KaufmanMass2017, DinurKaufman2017, KaufmanOppenheim2020}, states that
\begin{align}\label{eq:alev-lau}
    \gamma\wrapp{P_{n}^{\bigtriangledown}} \geq \frac{1}{n} \prod_{k=0}^{n-2} \min_{\tau \in X(k)} \gamma\wrapp{P^{\tau}}.
\end{align}
Given this, all that remains is to lower bound the $\gamma\wrapp{P^{\tau}}$. Typically, one can obtain extremely sharp control on $\gamma\wrapp{P^{\tau}}$ via explicit calculation when $\tau$ has codimension exactly $2$, since $P^{\tau}$ is a very simple Markov chain. For example, in the partite setting where $\abs{U_{i}} = 2$ for all $i=1,\dots,n$ (corresponding to when $\mu$ is a measure over the Boolean cube), $P^{\tau}$ is a Markov chain on at most $4$ states whose states correspond to pairs $(i,s)$ where $i$ is a coordinate not fixed by $\tau$ and $s$ is a possible assignment to coordinate $i$. Remarkably, trickle-down-type theorems allow one to control $\gamma\wrapp{P^{\tau}}$ for all $\tau$ using only information about the codimension-$2$ case. This can then be combined with \cref{eq:alev-lau} to obtain a global spectral gap for $P_{n}^{\bigtriangledown}$ based solely on local information.

\begin{theorem}\label{thm:oppenheim}
Suppose $P^{\tau}$ is irreducible for all $\tau \in X$ of codimension at least $2$ and that there exists $\gamma > 1 - \frac{1}{n-1}$ such that $\gamma\wrapp{P^{\tau}} \geq \gamma$ for all $\tau \in X(n-2)$. Then for every $2 \leq k \leq n$ and every $\tau \in X(n-k)$,
\begin{align*}
    \gamma\wrapp{P^{\tau}} \geq \frac{1 - (k-1) \cdot (1 - \gamma)}{1 - (k-2) \cdot (1 - \gamma)}.
\end{align*}
Moreover,
\begin{align*}
    \gamma\wrapp{P_{n}^{\bigtriangledown}} \geq \frac{1 - (n-1) \cdot (1 - \gamma)}{n}.
\end{align*}
\end{theorem}
The first claim is a result due to \cite{Oppenheim2018}; see also \cite{Gar73, BS97, Zuk03}. The second claim regarding the spectral gap of $P_{n}^{\bigtriangledown}$ is an immediate consequence of the first claim combined with \cref{eq:alev-lau}.

The second trickle-down-type theorem we will (re)prove concerns partite complexes specifically. For a face $\tau$, write $V(\tau)=\{i\in[n]:\tau\cap U_i\neq\emptyset\}$ for the parts met by $\tau$. For two distinct $u,v \in [n]$, define the \emph{spectral influence} of $u$ on $v$ by the symmetric quantity
\begin{align*}
    \mathcal{I}(u,v) = \mathcal{I}(v,u) \defeq \max_{\substack{\tau \in X(n-2) \\ u,v \notin V(\tau)}} \lambda_{2}\wrapp{P^{\tau}},
\end{align*}
with the convention that $\mathcal{I}(u,u) = 0$ for all $u \in [n]$. To avoid degeneracies, we only consider $\tau$ such that the conditional measure $\mu^{\tau}$ is not a point mass, and interpret an empty maximum as $0$. We collect these quantities into a symmetric matrix $\mathcal{I} \in \R^{n \times n}$.

\begin{theorem}\label{thm:LO26-improved}
Let $(X,\mu)$ be a pure $n$-partite complex, and assume $P^{\tau}$ is irreducible for all $\tau \in X$ of codimension at least $2$. Suppose there exists $0 < \epsilon < 1$ such that $\lambda_{\max}(\mathcal{I}) \leq 1 - \epsilon$. Then for every $2 \leq k \leq n$ and every $\tau \in X(n-k)$,
\begin{align*}
    \gamma\wrapp{P^{\tau}} \geq 1 - \frac{1 - \epsilon}{1 + \epsilon \cdot (k-2)}.
\end{align*}
Moreover,
\begin{align*}
    \gamma\wrapp{P_{n}^{\bigtriangledown}} \geq \frac{\epsilon}{n}.
\end{align*}
\end{theorem}
The first claim of \cref{thm:LO26-improved} was already established in \cite{LeakeOveisGharan2025}, building on prior work of \cite{AbdolazimiOveisGharan2023, LeakeLindbergOveisGharan2025}. Their analysis gives a global spectral-gap bound of the form $\gamma\wrapp{P_{n}^{\bigtriangledown}} \geq n^{-O(1/\epsilon)}$ \cite{LeakeOveisGharan2025}. Leake--Oveis Gharan left open the question of whether $\gamma\wrapp{P_{n}^{\bigtriangledown}}$ can be lower bounded instead by a fixed polynomial in $n$ multiplied by a function of $\epsilon$ that is independent of $n$. \cref{thm:LO26-improved} answers this affirmatively with essentially optimal parameters, and bypasses the theory of Lorentzian polynomials on cones \cite{BL26, LeakeLindbergOveisGharan2025, LeakeOveisGharan2025}.

\paragraph{Technical Overview} Our reproofs of \cref{thm:oppenheim,thm:LO26-improved} both proceed via an integrated Bochner method, which establishes a spectral gap lower bound of $\gamma$ on an irreducible Markov chain $P$ with stationary distribution $\mu$ by certifying the following spectral inequality in the inner product space $L^{2}(\mu)$:
\begin{align}\label{eq:bochner}
    \+L^{2} - \gamma \+L \succeq 0,
\end{align}
where $\+L = \id - P$ is the \emph{Laplacian} associated to $P$. Our approach to \cref{eq:bochner} for down-up walks is strongly inspired by \cite{KKO13, GoebelEtAl2026}, the latter of which uses an integrated Bochner method to reprove rapid mixing of Glauber dynamics for sampling from the hardcore model on random $d$-regular graphs beyond the uniqueness threshold. This was originally achieved in \cite{CCCYZ25} using a trickle-down-type argument for the ``negative fields localization scheme''.

For the down-up walk $P_{n}^{\bigtriangledown}$, in the spirit of Garland's method \cite{Gar73}, we observe that its Laplacian $\+L$ is $1/n$ times a sum of projection matrices corresponding to conditioning on codimension-$1$ faces. We can then group these projection matrices into blocks indexed by codimension-$2$ faces. With this decomposition in hand, we then expand the square $\+L^{2}$ and lower bound the cross terms by invoking the assumed bounds on the spectral gaps of the $P^{\tau}$ for $\tau$ with codimension-$2$. At an informal level, these local spectral gaps can be understood as a kind of ``curvature'' condition for high-dimensional simplicial complexes. Notably, we directly establish a spectral gap for $P_{n}^{\bigtriangledown}$ from the stated assumptions without invoking \cref{eq:alev-lau} nor spectral expansion of $P^{\tau}$ for $\tau$ with codimension larger than $2$. That the latter is implied by the former follows from the observation that a spectral gap for $P^{\emptyset}$ can be viewed as a Poincar\'{e} Inequality for $P_{n}^{\bigtriangledown}$ but restricted to the class of linear functions, see \cref{lem:universality}.

%\begin{remark}[Local walks via universality of spectral independence]\label{rmk:local-walks}
%    In the proofs below, it suffices to establish the global down-up spectral-gap bounds.
%    The local-walk bound in \cref{thm:LO26-improved} then follows by applying the same argument to each link and invoking the universality of spectral independence \cite[Section~3.1]{AnariEtAl2024}.
%    The sharper local-walk bound in \cref{thm:oppenheim} follows by adapting the universality argument of \cite{AnariEtAl2024} to retain the conditional-mean term and bound it using Cauchy--Schwarz.
%    We omit the detailed argument for these local-walk bounds, since the main purpose of this note is to establish the spectral gap of the global down-up walk.
%\end{remark}

\paragraph{Independent Work} We note that \cref{thm:LO26-improved} was independently proved by \cite{GuoZhang2026Planar} via an extremely similar argument. We discuss the slight difference in proof technique in \cref{rmk:GZ26-comparison} below. An even more recent paper of \cite{BLO26}, posted while we were writing this note, also recovers the same result in a more general framework.

\paragraph{AI Disclosure} The proofs were developed through a couple of rounds of interaction with GPT-5.6 Sol Ultra. ChatGPT was also used during the writing process to refine phrasing and provide editorial suggestions. The authors independently verified all mathematical claims, digested the proofs, wrote the bulk of the paper by hand, and take full responsibility for the final content.

\subsection{Preliminaries}
In this paper, every Markov chain will be reversible with respect to its (unique) stationary distribution $\mu$ and hence, self-adjoint with respect to the inner product $\langle f,g \rangle_{\mu} \defeq \E[x \sim \mu]{f(x)g(x)}$. The Loewner order $A \preceq B$ for self-adjoint operators $A,B$ is then defined with respect to this inner product, i.e. it means $\langle f, Af \rangle_{\mu} \leq \langle f, Bf \rangle_{\mu}$ for all $f \in L^{2}(\mu)$.

\begin{lemma}\label{lem:bochner}
Let $P$ be an irreducible Markov chain that is reversible with respect to its stationary distribution $\mu$, and let $\+L = \id - P$ be its Laplacian. Then $P$ has spectral gap at least $\gamma$ if and only if $\+L^{2} - \gamma \+L \succeq 0$.
\end{lemma}
\begin{proof}
The statement that $P$ has spectral gap at least $\gamma$ is equivalent to the spectral inequality $\+L \succeq \gamma \cdot \id$ restricted to the subspace $\allone^{\perp}$ endowed with inner product $\langle \cdot,\cdot \rangle_{\mu}$. On this inner product space, $\+L$ is positive definite by irreducibility and hence, is invertible. Hence, on this subspace, we have that $\+L \succeq \gamma \cdot \id$ is equivalent to $\+L^{2} \succeq \gamma \+L$ by conjugating by $\+L^{1/2}$ or $\+L^{-1/2}$ appropriately.
\end{proof}
\begin{remark}
We note that irreducibility is essential.
\end{remark}

\begin{lemma}\label{lem:universality}
Suppose $0<\epsilon\leq n$ and $\gamma\wrapp{P_{n}^{\bigtriangledown}} \geq \frac{\epsilon}{n}$. Then $\gamma\wrapp{P^{\emptyset}} \geq 1 - \frac{1-\epsilon}{1 + \epsilon \cdot (n-2)}$.
\end{lemma}
\begin{remark}
\cite[Section~3.1]{AnariEtAl2024} proved the above with denominator $\epsilon(n-1)$ in place of $1+\epsilon(n-2)$. The above gives a slight sharpening that makes the two claims of \cref{thm:oppenheim} quantitatively consistent with each other. A short proof is provided in \cref{sec:universality-proof} to keep this note self-contained.
\end{remark}

\section{On Oppenheim's Trickle-Down Theorem}\label{sec:oppenheim}
In this section, we reprove \cref{thm:oppenheim}. Throughout this section, let $(X,\mu)$ be a pure weighted simplicial complex whose facets have cardinality $n$; we do not assume $X$ is $n$-partite. For every codimension-$1$ face $\tau \in X(n-1)$, define the conditional expectation operator
\begin{align*}
    (Q_\tau f)(\sigma) \defeq \*1[\sigma\supseteq \tau]\cdot \E[\eta\sim \mu^{\tau}]{f(\eta)} + \*1[\sigma\not\supseteq \tau] \cdot f(\sigma).
\end{align*}
This is a matrix that encodes all the transitions of the down-up walk $P_{n}^{\bigtriangledown}$ that ``go through'' $\tau$. We then define its Laplacian $\+L_\tau \defeq \id - Q_\tau$, and let
\begin{align*}
    \+L \defeq \sum_{\tau \in X(n-1)} \+L_\tau.
\end{align*}
We have the following elementary facts.
\begin{observation}\label{obs:L-down-up}
For all $\tau \in X(n-1)$, we have $Q_\tau^2 = Q_\tau$ and $\+L_\tau^2 = \+L_\tau$. Furthermore, $\+L = n\tp{\id - P_n^{\bigtriangledown}}$.
\end{observation}
\begin{proof}
The first claim is immediate from the definition of $Q_{\tau}$. For the second claim, we have
\begin{align*}
    (\+L f)(\sigma)
    &= \sum_{\tau \in X(n-1)} (\+L_\tau f)(\sigma) \\
    &= \sum_{\tau \in X(n-1): \tau \subseteq \sigma} \tp{f(\sigma) - \E[\eta\sim \mu^{\tau}]{f(\eta)}} \\
    &= n\cdot f(\sigma) - n \cdot \frac{1}{n}\sum_{\tau \in X(n-1): \tau \subseteq \sigma} \E[\eta\sim \mu^{\tau}]{f(\eta)} \\
    &= n \wrapp{\wrapp{\id - P_n^{\bigtriangledown}} f}(\sigma),
\end{align*}
where the second equation follows by $(\+L_\tau f)(\sigma) = 0$ for $\sigma\not\supseteq \tau$.
\end{proof}

Next, we decompose $\+L^{2} - \+L$ in the spirit of Garland's method \cite{Gar73}.
\begin{proposition} \label{prop:garland}
    Suppose $n \geq 2$. For every $\alpha \in X(n-2)$, define
    \begin{align*}
        \+L_{\alpha} \defeq \sum_{\tau \in X(n-1) : \tau \supseteq \alpha} \+L_\tau.
    \end{align*}
    Then, it holds that 
    \begin{align*}
        \+L^2 - \+L = \sum_{\alpha \in X(n-2)} \tp{\+L_{\alpha}^2 - \+L_{\alpha}}.
    \end{align*}
\end{proposition}
\begin{proof}
    By definition of $\+L$, we have
    \begin{align*}
        \+L^2 - \+L
        &= \sum_{\substack{\sigma,\tau \in X(n-1) \\ \abs{\sigma \oplus \tau} \geq 2}} \+L_\sigma \+L_\tau
        = \sum_{\substack{\sigma,\tau \in X(n-1) \\ \abs{\sigma \oplus \tau} = 2}} \+L_\sigma \+L_\tau,
    \end{align*}
    where the last equation follows by the fact that if $\abs{\sigma\oplus \tau} \geq 3$, then there is no $\eta \in X(n)$ such that both $\eta\supseteq \sigma$ and $\eta \supseteq \tau$ hold. Now, note that for every $\sigma,\tau \in X(n-1)$ such that $\abs{\sigma\oplus\tau} = 2$, there is a unique $\alpha \in X(n-2)$ such that $\sigma,\tau \supseteq \alpha$. Hence, it is straightforward to regroup the right hand side by $\alpha \in X(n-2)$ and finish the proof.
\end{proof}

\begin{proposition} \label{prop:local-expansion}
    Let $\alpha \in X(n-2)$. Assume the down-up walk on $\mu^{\alpha}$ is irreducible.
    The following inequality holds in $L^2(\mu)$ (and restricts to $L^2(\mu^{\alpha})$),
    \begin{align*}
        \+L_{\alpha}^2 \succeq \gamma\wrapp{P^{\alpha}} \cdot \+L_{\alpha}.
    \end{align*}
\end{proposition}
\begin{proof}
    The block operator vanishes outside facets containing $\alpha$, so we focus on $(X^{\alpha}, \mu^{\alpha})$ and the space $L^2(\mu^{\alpha})$. If $\mu^{\alpha}$ is a point mass, the claim is trivial. Now, let $P_2^{\bigtriangledown}$ denote the down-up walk on this specific link $(X^{\alpha},\mu^{\alpha})$. By \Cref{obs:L-down-up}, $\+L_{\alpha} = 2(\id - P_2^\bigtriangledown)$ and hence, the target inequality is equivalent to 
    \begin{align*}
        \wrapp{\id - P_2^\bigtriangledown}^{2} \succeq \frac{\gamma\wrapp{P^{\alpha}}}{2} \cdot \wrapp{\id - P_2^\bigtriangledown}.
    \end{align*}
    Since the down-up walk is irreducible, by \cref{lem:bochner}, it suffices to show that
    \begin{align*}
        \lambda_2\wrapp{P_2^\bigtriangledown} &= \frac{1}{2} + \frac{1}{2} \lambda_2(P^{\alpha}).
    \end{align*}
    In fact, $P_{2}^{\bigtriangledown}$ and $\frac{1}{2} \cdot \wrapp{\id + P^{\alpha}}$ have the exact same spectrum up to the multiplicity of the zero eigenvalue. Indeed, let $D^{\alpha}$ denote the transition matrix that encodes the following action: Given $\sigma \in X(n)$ satisfying $\sigma \supseteq \alpha$, remove a uniformly random element in $\sigma \setminus \alpha$ from $\sigma$ to obtain $\tau \in X(n-1)$. Let $U^{\alpha}$ denote the transition matrix of the adjoint action: Given $\tau \in X(n-1)$ satisfying $\tau \supseteq \alpha$, sample a random element from $\mu_{1}^{\tau}$ and add it to $\tau$ to obtain $\sigma \in X(n)$. It is straightforward to verify that $P_{2}^{\bigtriangledown} = D^{\alpha}U^{\alpha}$ and $\frac{1}{2} \cdot \wrapp{\id + P^{\alpha}} = U^{\alpha}D^{\alpha}$.
\end{proof}

\begin{proof}[Proof of \cref{thm:oppenheim}]
    Let us focus first on lower bounding the spectral gap of the global down-up walk $P_{n}^{\bigtriangledown}$. Irreducibility of all local walks $P^{\tau}$ ensures irreducibility of $P_{n}^{\bigtriangledown}$. Hence, by \cref{lem:bochner} and \cref{obs:L-down-up}, it suffices to show that $\+L^2 \succeq (1 - (n-1) \cdot (1 - \gamma)) \cdot \+L$. For this, we have
    \begin{align*}
        \+L^2 - \+L
        &= \sum_{\alpha \in X(n-2)} \tp{\+L_{\alpha}^2 - \+L_{\alpha}} \tag{\cref{prop:garland}} \\
        &\succeq -(1 - \gamma) \sum_{\alpha \in X(n-2)} \+L_{\alpha} \tag{\cref{prop:local-expansion}} \\
        &= -(n-1) \cdot (1 - \gamma) \cdot \+L,
    \end{align*}
    where in the last equation, we use the fact that for every $\sigma \in X(n-1)$, there are exactly $(n-1)$ distinct $\alpha \in X(n-2)$ with $\alpha \subseteq \sigma$. This proves $\gamma\wrapp{P_{n}^{\bigtriangledown}} \geq \frac{1 - (n-1) \cdot (1 - \gamma)}{n}$ as desired.

    Regarding the first claim on spectral gaps of local walks, fix $\tau \in X(n-k)$ where $2 \leq k \leq n$. Applying the preceding argument mutatis mutandis to the link $(X^{\tau},\mu^{\tau})$ shows that its global down-up walk has spectral gap at least $\frac{1 - (k-1) \cdot (1 - \gamma)}{k}$. Applying \cref{lem:universality} with $\epsilon=1-(k-1)(1-\gamma)\in(0,k]$ then completes the proof.
\end{proof}

\section{Rapid Mixing via Pairwise Spectral Influence}\label{sec:LO25}
In this section, we reprove \cref{thm:LO26-improved}. Throughout this section, let $(X,\mu)$ denote a pure weighted $n$-partite simplicial complex. For each $u\in[n]$, let $P_u$ be the heat-bath projection that resamples coordinate $u$:
\begin{align*}
    (P_u f)(\sigma) \defeq \E[\eta\sim\mu^{\sigma\setminus U_u}]{f(\eta)},
    \qquad \+L_u \defeq \id-P_u.
\end{align*}
The unnormalized Laplacian from \cref{sec:oppenheim} is
$\+L=\sum_{u\in[n]}\+L_u=n(\id-P_n^{\bigtriangledown})$. The key lemma is the following.
\begin{lemma} \label{lem:LO26-local-mixing}
  For every $u \neq v$ in $[n]$, every $\sigma \in X(n-2)$ such that $u,v \not\in V(\sigma)$, and every function $f$,
  \begin{align*}
    \inner{\+L_u f}{\+L_v f}_{\mu^{\sigma}} \geq -\lambda_2(P^{\sigma}) \cdot \norm{\+L_u f}_{\mu^{\sigma}, 2} \cdot \norm{\+L_v f}_{\mu^{\sigma}, 2}.
  \end{align*}
\end{lemma}

\begin{corollary}\label{cor:LO26-local-mixing}
For every $u \neq v$ in $[n]$ and every function $f$,
\begin{align*}
      \inner{\+L_{u} f}{\+L_{v} f}_{\mu} \geq -\mathcal{I}(u,v) \cdot \norm{\+L_{u}f}_{\mu,2} \cdot \norm{\+L_{v}f}_{\mu,2}.
  \end{align*}
\end{corollary}
\begin{proof}
Averaging over $\sigma=\eta\setminus(U_u\cup U_v)$ for $\eta\sim\mu$, we have
  \begin{align*}
    \inner{\+L_u f}{\+L_v f}_{\mu}
    &= \E[\sigma]{\inner{\+L_u f}{\+L_v f}_{\mu^{\sigma}}} \tag{Law of Total Probability} \\
    &\geq -\mathcal{I}(u,v) \cdot \E[\sigma]{\norm{\+L_u f}_{\mu^{\sigma}, 2} \cdot \norm{\+L_v f}_{\mu^{\sigma}, 2}} \tag{\cref{lem:LO26-local-mixing}} \\
    &\geq -\mathcal{I}(u,v) \cdot \norm{\+L_u f}_{\mu, 2} \cdot \norm{\+L_v f}_{\mu, 2}. \tag{Cauchy--Schwarz and $\mathcal{I}(u,v) \geq 0$}
  \end{align*}
\end{proof}

We first use this to complete the proof of \cref{thm:LO26-improved}.
\begin{proof}[Proof of \cref{thm:LO26-improved}]
  We first bound the spectral gap of the global down-up walk. Observe that for any $f$,
  \begin{align*}
    \norm{\+L f}_{\mu,2}^2
    &= \sum_w \norm{\+L_w f}_{\mu,2}^2 + \sum_{u \neq v} \inner{\+L_u f}{\+L_v f}_{\mu} \\
    &\geq \sum_w \norm{\+L_w f}_{\mu,2}^2 - \sum_{u \neq v} \mathcal{I}(u,v) \cdot \norm{\+L_u f}_{\mu, 2} \cdot \norm{\+L_v f}_{\mu, 2} \tag{\cref{cor:LO26-local-mixing}} \\
    &= \*x^{\intercal} \wrapp{\id - \mathcal{I}} \*x \tag{where $\*x \in \R^{n}$ has $x_{u} = \norm{\+L_{u} f}_{\mu,2}$} \\
    &\geq \epsilon \*x^{\intercal} \*x \tag{$\lambda_{\max}(\mathcal{I}) \leq 1 - \epsilon$} \\
    &= \epsilon \inner{f}{\+L f}_{\mu}. \tag{Using $\+L_{u}^{2} = \+L_{u}$ for all $u \in [n]$}
  \end{align*}
  Thus $\+L^2\succeq\epsilon\+L$. By \cref{lem:bochner} and $\+L = n\wrapp{\id-P_n^{\bigtriangledown}}$, this proves $\gamma(P_n^{\bigtriangledown})\geq\epsilon/n$ as desired.

  Regarding the first claim on spectral gaps of local walks, fix $\tau \in X(n-k)$ where $2 \leq k \leq n$. For $u,v \notin V(\tau)$, the spectral influence $\mathcal{I}^{\tau}(u,v)$ after conditioning on $\tau$ is upper bounded by the unconditional spectral influence $\mathcal{I}(u,v)$. Hence, by monotonicity of the spectral radius for nonnegative matrices (extending $\mathcal{I}^{\tau}$ by zero on the fixed parts), $\lambda_{\max}(\mathcal{I}^{\tau}) \leq \lambda_{\max}(\mathcal{I}) \leq 1 - \epsilon$. Applying the preceding argument mutatis mutandis to the global down-up walk of the conditional measure $\mu^{\tau}$ then yields a spectral gap lower bound of $\epsilon/k$. \cref{lem:universality} then finishes the proof.
\end{proof}

\subsection{Proof of \texorpdfstring{\Cref{lem:LO26-local-mixing}}{Lemma 3.1}}
Fix $\sigma\in X(n-2)$ and distinct $u,v \notin V(\sigma)$. Throughout this subsection, we work in $(X^{\sigma}, \mu^{\sigma})$ and the associated inner product space $L^{2}(\mu^{\sigma})$; to simplify notation, we omit the dependence of the inner product (and the induced norm) on $\mu^{\sigma}$. Our goal is to establish the inequality
\begin{align} \label{eq:LO26-local-target}
  \inner{\+L_u f}{\+L_v f} \geq -\lambda_2(P^{\sigma}) \cdot \norm{\+L_u f} \cdot \norm{\+L_v f}
\end{align}
for all real-valued functions $f$ on the facets of the $2$-partite link $X^{\sigma}$. If $\mu^{\sigma}$ is a point mass, both $\+L_u$ and $\+L_v$ vanish, so the claim is trivial. If $P^{\sigma}$ is reducible, then $\lambda_2(P^{\sigma})=1$ and the claim follows from Cauchy--Schwarz. Otherwise, $P^{\sigma}$ is a walk on a bipartite graph with more than two vertices, so $\lambda_2(P^{\sigma})\geq0$. In this case, our first step will be to reduce \cref{eq:LO26-local-target} to the case of \emph{linear functions}, i.e. functions of the form $f=g+h$ where $g$ depends only on coordinate $v$ and $h$ only on coordinate $u$.
\begin{lemma} \label{lem:linear-func-sufficiency}
  If \cref{eq:LO26-local-target} holds for all linear functions, then it holds for all functions.
\end{lemma}

\begin{proof}
For an operator $Q$, we use $\Ran Q$ and $\Ker Q$ to denote the range and kernel of $Q$, respectively. Recalling that $P_{u},P_{v}$ are the heat-bath projections defined above, define $K \defeq \Ker P_u \cap \Ker P_v$. Since $P_u$ and $P_v$ are self-adjoint with respect to $\inner{\cdot}{\cdot}$, the subspaces $K, K^\perp$ are closed under both $P_u$ and $P_v$.

Now, let $f$ be an arbitrary function, and decompose it as $f = f^K + f^{\perp K}$ where $f^K \in K$ and $f^{\perp K} \in K^{\perp}$. If we can establish that every function in $K^{\perp}$ is linear, then
\begin{align*}
  \inner{\+L_u f}{\+L_v f}
  &= \inner{\+L_u f^K}{\+L_v f^K} + \inner{\+L_u f^{\perp K}}{\+L_v f^{\perp K}} \\
  &= \norm{f^K}^2 + \inner{\+L_u f^{\perp K}}{\+L_v f^{\perp K}} \\
  &\geq \inner{\+L_u f^{\perp K}}{\+L_v f^{\perp K}} \\
  &\geq -\lambda_{2}(P^{\sigma}) \cdot \norm{\+L_{u} f^{\perp K}} \cdot \norm{\+L_{v} f^{\perp K}} \tag{Linearity of $f^{\perp K}$} \\
  &\geq -\lambda_{2}(P^{\sigma}) \cdot \norm{\+L_{u} f} \cdot \norm{\+L_{v} f} \tag{Using $\lambda_{2}(P^{\sigma}) \geq 0$}
\end{align*}
completes the proof. To show that all functions in $K^{\perp}$ are linear, note that $K^\perp = \Ran P_u + \Ran P_v$. Hence, if $f \in K^\perp$, there exist functions $\tilde{g}, \tilde{h}$ such that
\begin{align*}
  f = P_{u}\tilde{g} + P_{v}\tilde{h} = \E{\tilde{g}(\eta)\mid \eta(v)} + \E{\tilde{h}(\eta) \mid \eta(u)}, \quad \eta \sim \mu^{\sigma}.
\end{align*}
Clearly, $g \defeq P_{u}\tilde{g}$ depends only on coordinate $v$ and $h \defeq P_{v}\tilde{h}$ depends only on coordinate $u$.
\end{proof}

Now we are ready to prove \cref{eq:LO26-local-target} and \Cref{lem:LO26-local-mixing}.
\begin{proof}[Proof of \Cref{lem:LO26-local-mixing}]
Noting that $\+L_{\sigma}=\+L_u+\+L_v$ on this link, \cref{prop:local-expansion} applied to the $2$-partite link $(X^{\sigma},\mu^{\sigma})$ gives
\begin{align*}
    \inner{f}{(\+L_u + \+L_v)^2 f} \geq (1-\lambda_2(P^{\sigma})) \cdot \inner{f}{(\+L_u + \+L_v)f}.
\end{align*}
Expanding $(\+L_u + \+L_v)^2$ and rearranging gives
\begin{align} \label{eq:Bochner-local-GD}
    \inner{\+L_uf}{\+L_v f} \geq -\lambda_2(P^{\sigma}) \cdot \frac{\norm{\+L_u f}^2 + \norm{\+L_v f}^2}{2},
\end{align}
where we used $\+L_u^2 = \+L_u, \+L_v^2 = \+L_v$ and self-adjointness. We would like to replace the arithmetic mean in the right-hand side with the geometric mean to conclude \cref{eq:LO26-local-target}. By \Cref{lem:linear-func-sufficiency}, we may assume without loss of generality that $f = g + h$, where $g, h$ depend only on coordinates $v$ and $u$, respectively. Noting that $P_u g = g, P_v h = h$ and $\+L_u g = 0, \+L_v h = 0$, \cref{eq:Bochner-local-GD} is equivalent to
\begin{align*}
    \inner{\+L_uh}{\+L_v g} \geq -\lambda_2(P^{\sigma}) \cdot \frac{\norm{\+L_u h}^2 + \norm{\+L_v g}^2}{2}.
\end{align*}
Since this holds for any $g,h$, we have that for every $t > 0$,
\begin{align*}
    \inner{\+L_uf}{\+L_v f} = \inner{\+L_uh}{\+L_v g} 
    &= \inner{\+L_u (th)}{\+L_v (g/t)} \\
    &\geq -\lambda_2(P^{\sigma}) \cdot \frac{t^2 \norm{\+L_u h}^2 + \norm{\+L_v g}^2/t^2}{2} \\
    &= -\lambda_2(P^{\sigma}) \cdot \frac{t^2 \norm{\+L_u f}^2 + \norm{\+L_v f}^2/t^2}{2}.
\end{align*}
  If either norm is zero, the claim is immediate. Otherwise, taking $t^2 = \norm{\+L_v f}/\norm{\+L_u f}$ finishes the proof.
\end{proof}

\begin{remark}[Comparison with Guo--Zhang]\label{rmk:GZ26-comparison}
    The proof of \cite[Lemma~6]{GuoZhang2026Planar} uses the same expansion of $\+L^2$ and the same bounds on the cross terms as in our proof of \cref{thm:LO26-improved}.
    The slight difference lies in the proof of the two-coordinate estimate \cref{lem:LO26-local-mixing}: \cite[Lemma~5]{GuoZhang2026Planar} is derived using the Hirschfeld--Gebelein--R\'{e}nyi maximal correlation and positive semidefiniteness of a $3\times3$ Gram matrix associated with $f,P_u f,P_v f$.
    Our proof first reduces to linear functions and the two-coordinate Bochner-type inequality stated in \cref{prop:local-expansion}.
\end{remark}

{\sloppy
\printbibliography
}

\appendix

\section{Proof of \texorpdfstring{\cref{lem:universality}}{the universality lemma}}\label{sec:universality-proof}
The rough intuition is that the local walk $P^{\emptyset}$ exactly encodes pairwise correlations in $\mu$ (see e.g. \cite{ALO2020}), and that these pairwise correlations can be studied by applying \emph{linear} test functions to the global spectral gap of $P_{n}^{\bigtriangledown}$. Write $\lambda=\lambda_2(P^{\emptyset})$; our goal is to show $\lambda\leq(1-\epsilon)/(1+\epsilon(n-2))$. Let $v \neq 0$ be a corresponding eigenvector orthogonal to constants, so $\E[x \sim \mu_{1}]{v_{x}} = \langle v, \allone \rangle_{\mu_1} = 0$. We lift this to a global function on facets via $f(\sigma) \defeq \sum_{x \in \sigma} v_{x}$ for $\sigma \in X(n)$, and apply this test function $f$ to the Poincar\'{e} Inequality
\begin{align*}
    \frac{\epsilon}{n} \cdot \Var[\mu]{f} \leq \langle f, \wrapp{\id - P_{n}^{\bigtriangledown}}f \rangle_{\mu} = \E[\tau \sim \mu_{n-1}]{\Var[\mu^{\tau}]{f}},
\end{align*}
which is well-known to be equivalent to the assumed spectral gap lower bound $\gamma\wrapp{P_{n}^{\bigtriangledown}} \geq \frac{\epsilon}{n}$. Since $\E[x \sim \mu_{1}]{v_{x}} = 0$, we have $\E[\sigma \sim \mu]{f(\sigma)} = 0$ by the definition of $f$. The variance on the left-hand side is
\begin{align*}
    \Var[\mu]{f} &= \E[\sigma \sim \mu]{f(\sigma)^{2}} \\
    &= \sum_{x} \Pr[\sigma \sim \mu]{x \in \sigma} \cdot v_{x}^{2} + \sum_{x \neq y} \Pr[\sigma \sim \mu]{x,y \in \sigma} \cdot v_{x}v_{y} \\
    &= n \cdot \langle v,v \rangle_{\mu_{1}} + n \cdot (n-1) \cdot \langle v, P^{\emptyset}v \rangle_{\mu_{1}} \\
    &= n \cdot (1 + (n-1) \cdot \lambda) \cdot \norm{v}_{\mu_{1}}^{2}.
\end{align*}
Nonnegativity of the variance gives $\lambda\geq-1/(n-1)$, so $1+(n-2)\lambda>0$.
For the right-hand side, for $\tau \in X(n-1)$, again define $f(\tau) = \sum_{x \in \tau} v_{x}$ and write $g(\tau) = \E[x \sim \mu_{1}^{\tau}]{v_{x}}$. Then
\begin{align*}
    \E[\tau \sim \mu_{n-1}]{\Var[\mu^{\tau}]{f}} &= \E[\tau \sim \mu_{n-1}]{\E[x \sim \mu_1^\tau]{v_x^2}} - \E[\tau \sim \mu_{n-1}]{g(\tau)^{2}}.
\end{align*}
The first term equals $\norm{v}_{\mu_{1}}^{2}$. The second term is clearly nonnegative, and ending the analysis here would recover the result of \cite[Section~3.1]{AnariEtAl2024}, with denominator $\epsilon(n-1)$ in place of $1+\epsilon(n-2)$. We will sharpen the analysis by lower bounding the second term by a multiple of $\norm{v}_{\mu_{1}}^{2}$. For this, observe that by the same argument as above,
\begin{align*}
    \E[\tau \sim \mu_{n-1}]{f(\tau)^{2}} &= (n-1) \cdot (1 + (n-2) \cdot \lambda) \cdot \norm{v}_{\mu_{1}}^{2}
\end{align*}
and
\begin{align*}
    \E[\tau \sim \mu_{n-1}]{f(\tau)g(\tau)} &= (n-1) \cdot \E[\tau \sim \mu_{n-1}]{\E[x \sim \mathrm{Unif}(\tau)]{v_{x}} \cdot \E[y \sim \mu_{1}^{\tau}]{v_{y}}} \\
    &= (n-1) \cdot \E[x \sim \mu_{1}]{v_x\cdot\E[\substack{\tau \sim \mu_{n-1} \mid \tau \ni x \\ y \sim \mu_{1}^{\tau}}]{v_{y}}} \\
    &= (n-1) \cdot \langle v, P^{\emptyset}v \rangle_{\mu_{1}} \\
    &= (n-1) \cdot \lambda \cdot \norm{v}_{\mu_{1}}^{2}.
\end{align*}
By Cauchy--Schwarz,
\begin{align*}
    \E[\tau \sim \mu_{n-1}]{g(\tau)^{2}} \geq \frac{\E[\tau \sim \mu_{n-1}]{f(\tau)g(\tau)}^{2}}{\E[\tau \sim \mu_{n-1}]{f(\tau)^{2}}} = \frac{(n-1) \cdot \lambda^{2}}{1 + (n-2) \cdot \lambda} \cdot \norm{v}_{\mu_{1}}^{2}.
\end{align*}
Combining these estimates gives
\begin{align*}
    \epsilon\bigl(1+(n-1)\lambda\bigr)
    &\leq 1-\frac{(n-1)\lambda^2}{1+(n-2)\lambda}
    =\frac{(1-\lambda)\bigl(1+(n-1)\lambda\bigr)}{1+(n-2)\lambda}.
\end{align*}
If $1+(n-1)\lambda=0$, the desired bound follows directly from $\epsilon\leq n$. Otherwise, cancelling this positive factor and rearranging yields
\begin{align*}
    \lambda\leq\frac{1-\epsilon}{1+\epsilon(n-2)},
\end{align*}
which completes the proof.

\end{document}